\documentclass{article}
\usepackage{amsmath}
\usepackage{graphicx}
\usepackage[margin=1.3in]{geometry}

\usepackage{newtxtext}
\usepackage[subscriptcorrection]{newtxmath}

\usepackage{natbib}

\graphicspath{{./figures/}}

\usepackage[plain,noend]{algorithm2e}

\makeatletter
\renewcommand{\algocf@captiontext}[2]{#1\algocf@typo. \AlCapFnt{}#2} 
\def\@algocf@capt@plain{top}
\renewcommand{\algocf@makecaption}[2]{%
  \addtolength{\hsize}{\algomargin}%
  \sbox\@tempboxa{\algocf@captiontext{#1}{#2}}%
  \ifdim\wd\@tempboxa >\hsize
    \hskip .5\algomargin%
    \parbox[t]{\hsize}{\algocf@captiontext{#1}{#2}}
  \else%
    \global\@minipagefalse%
    \hbox to\hsize{\box\@tempboxa}
  \fi%
  \addtolength{\hsize}{-\algomargin}%
}
\makeatother

\usepackage{soul}
\usepackage{color}
\usepackage{url}
\usepackage[normalem]{ulem}

\newtheorem{theorem}{Theorem}

\newcommand{\taux}{T_{\text{aux}}}
\newcommand{\E}{\mathrm{E}}
\newcommand{\var}{\mathrm{var}}
\newcommand{\avar}{\mathrm{avar}}

\newcommand{\pr}{\mathrm{pr}}
\newcommand{\Real}{\mathbb{R}}
\newcommand{\T}{^{ \mathrm{\scriptscriptstyle T} }} 
\newcommand{\dnorm}{{N}}

\newcommand{\sumdot}{\text{\tiny$\bullet$}}
\newcommand{\nid}{N_{i\sumdot}}
\newcommand{\ndj}{N_{\sumdot j}}

\newcommand{\yij}{Y_{ij}}
\newcommand{\zij}{Z_{ij}}

\newcommand{\ai}{a_i}
\newcommand{\bj}{b_j}
\newcommand{\giv}{\mid}
\newcommand{\wh}{\hat} 

\newcommand{\rd}{\,\mathrm{d}}  
\newcommand{\err}{\varepsilon}

\newcommand{\sign}{\mathrm{sign}}

\newcommand{\ssa}{\sigma^2_A}
\newcommand{\ssb}{\sigma^2_B}
\newcommand{\siga}{\sigma_A}
\newcommand{\sigb}{\sigma_B}

\newcommand{\obset}{\mathcal{S}} 

\newcommand{\phm}{\phantom{-}} 

\newcommand{\Sooo}{\textsc{Bal-Nul-Hi}}
\newcommand{\Soot}{\textsc{Bal-Nul-Lo}}
\newcommand{\Soto}{\textsc{Bal-Lin-Hi}}
\newcommand{\Sott}{\textsc{Bal-Lin-Lo}}
\newcommand{\Stoo}{\textsc{Imb-Nul-Hi}}
\newcommand{\Stot}{\textsc{Imb-Nul-Lo}}
\newcommand{\Stto}{\textsc{Imb-Lin-Hi}}
\newcommand{\Sttt}{\textsc{Imb-Lin-Lo}}

\begin{document}

\title{Consistent and Scalable Composite Likelihood Estimation of Probit Models with Crossed Random Effects}

\author{
R. Bellio\\
University of Udine
\and S. Ghosh\\
Stanford University
\and  A. B. Owen\\
Stanford University
\and C. Varin\\
Ca' Foscari University
}
\date{April 2025}

\maketitle

\begin{abstract}
Estimation of crossed random effects models commonly requires computational costs that grow faster than linearly in the sample size $N$, often as fast as $\Omega(N^{3/2})$, making them unsuitable for large data sets.  For non-Gaussian responses, integrating out the random effects to get a marginal likelihood brings significant challenges, especially for high dimensional integrals where the Laplace approximation might not be accurate.  
We develop a composite likelihood approach to probit models that replaces the crossed random effects model with some hierarchical models that require only one-dimensional integrals.  We show how to consistently estimate the crossed effects model parameters from the hierarchical model fits.
We find that the computation scales linearly in the sample size. We illustrate the method on about five million observations from Stitch Fix where the crossed effects formulation would require an integral of dimension larger than $700{,}000$.
\end{abstract}


\section{Introduction}\label{sect:intro}
 
In this paper we develop a new composite likelihood approach to handling probit models with two crossed random effects.  Our initial motivation was to get point and interval parameter estimates at a computational cost that grows only linearly in the sample size, $N$. Standard algorithms for crossed random effects typically have superlinear cost, commonly $\Omega(N^{3/2})$, making them unsuitable for modern large data sets. A second issue is that the marginal likelihood in a crossed random effects model is an integral over $\Real^D$ where $D$ is large enough to make the integration problem challenging. Our scalable estimation method 
replaces this $D$-dimensional integral by $D$ integrals of dimension one. 

The common notation for mixed effects models combines fixed and random effects through a formula such as $X{\beta} + Z{b}$
using matrices $X$ and $Z$ of known predictors with unknown coefficient vectors ${\beta}$
and ${b}$ where ${b}$ happens to be random. This formulation is simple and elegant but it hides some extremely important
practical differences.  
As discussed above, the crossed setting leads to one high-dimensional integral while the hierarchical one uses many low-dimensional integrals.  For Gaussian responses we can use generalized least squares  on the response vector, without explicitly solving an integral.  
Even there, the crossed setting is harder. 
A hierarchical model has a block diagonal covariance matrix for the response vector resulting in a linear cost. For unbalanced crossed random effects, generalized least squares typically has a superlinear cost.

With the size $N$ of data sets growing rapidly
it is not possible to use estimation methods with a cost of $\Omega(N^{3/2})$.
Ideally the cost should be $O(N)$. 
We are motivated by electronic commerce
problems with large data sets.  A company might have customers $i=1,\dots,R$,
to which it sells items $j=1,\dots,C$. It would then be interested in modeling
how a response $Y_{ij}$ depends on some predictors ${x}_{ij}\in\Real^p$.
If they do not account for the fact that $Y_{ij}$ and $Y_{is}$  are correlated
due to a common customer $i$ or that $Y_{ij}$ and $Y_{rj}$ are correlated
due to a common item $j$, they will get an inefficient estimate. This flaw is less severe when $N$ is so large. What is very serious is that they will get unreliable standard errors for their estimates. 
In a setting where accounting for random effects is computationally impossible we expect that many users will simply ignore them, getting very na\"ive variance estimates and finding too many things significant.

A typical feature of data in our motivating applications 
is very sparse and imbalanced sampling. Only $N\ll RC$ of the possible $({x}_{ij},Y_{ij})$ values
are observed. There is generally no simple structure in the
pattern of which $(i,j)$ pairs are observed. It is common for the data to have very unequal sampling frequencies in each of the row and column variables.

The scaling problem is easiest to describe for generalised least squares solutions to linear mixed models, based on results from \cite{gao:owen:2020}. They note that the algorithm for generalised least squares involves solving a system of $R+C$ equations in $R+C$ unknowns which has cost $\Omega\{(R+C)^3\}$ in standard implementations. If $RC>N$ then $\max(R,C)>N^{1/2}$ and so $(R+C)^3>N^{3/2}$. The average number of observations per level is $N/(RC)$. This ratio is well below one in our motivating applications, and as long as it is $o(N^{1/3})$, the cost of standard algebra will be superlinear.

Standard Bayesian solutions run into a similar difficulty. For an intercept plus crossed random effects model, \cite{gao:2017} show that the Gibbs sampler takes $\Omega(N^{1/2})$ iterations that each have a cost proportional to $N$, for $\Omega(N^{3/2})$ cost overall. Several other Bayesian approaches considered there also had difficulty.

There has been recent progress on scalable algorithms
for crossed random effects problems improving both the frequentist and
Bayesian approaches.  For regression problems \cite{ghosh:2020} replace standard equation solving by a backfitting algorithm that has $O(N)$ cost per iteration and they give conditions under which the number of  iterations to convergence is $O(1)$ as $N\to\infty$.   See also \cite{ghandwani:2023} for regression with random slopes. 
\cite{papaspiliopoulos:2020} use a collapsed Gibbs sampler and then give conditions under which it has linear cost in $N$ for the intercept-only crossed random effects regression model. \cite{ghoshzhong:2021} do the same after weakening a stringent balance assumption. 

We consider a binary response requiring a generalized linear mixed effects model with crossed random effects that incur the high dimensional integration problem mentioned above. There are fewer scalable solutions for this problem. There is a frequentist approach by \cite{ghosh:2021}  and a Bayesian approach by \cite{Papaspiliopoulos:2023}. 
The all-row-column probit given here is simpler than
those two approaches and uses {much} weaker sampling
assumptions.

\cite{ghosh:2021} developed a penalized quasi-likelihood approach to logistic regression on fixed effects and two crossed random effects. That approach had iterations costing $O(N)$ each and empirically the number of iterations was $O(1)$. They used estimating equations from \cite{bres:clay:1993} which were based on work by \cite{schall:1991} to maximize  the marginal likelihood using a Laplace approximation.
The quantity being estimated is not exactly the maximum likelihood estimate.  It is a posterior mode corresponding to a not very informative prior, using some plugged-in weights, a quantity that goes back to \cite{stir:lair:ware:1984}. Penalized quasi-likelihood has a bias that can prevent it from being consistent.  Even with just the intercept and one random effect, it requires the number of levels of that effect as well as the number of observations at each of those levels to diverge to infinity in order to get a consistent estimate. With sample sizes $R=N^\rho$ and $C=N^\kappa$, for $\rho,\kappa\in(0,1)$ they require
$\max(\rho+2\kappa,2\rho+\kappa)<2$ and the observation probability for $(x_{ij},\yij)$ can only vary over
a narrow range.

\cite{Papaspiliopoulos:2023} extend the collapsed Gibbs sampler to obtain scalable Bayesian inference in generalized linear mixed models with crossed random effects using a reparameterization called local centering.  
They include an intercept and $K\ge2$ crossed random effects, while also discussing how fixed effects could be incorporated. 
Their approach requires a stringent balance assumption. For our data we would need $\nid =N/C$ for all rows $i=1,\dots,R$ and $\ndj=N/R$ for all $j=1,\dots,C$. With this condition, their cost per iteration is $O(N+R+C)$ in our notation. The total cost is this cost per iteration times a relaxation time.  For $K=2$ random effects and a discrete response like we have, they introduce an auxiliary relaxation time $\taux$ and show in Theorem 1 that the cost is $O[\max(N,R+C)\min\{2N/(R+C),\taux\}]$. Our problem has $\max(R,C)\ll N$ and then their cost is linear in $N$ if and only if $\taux=O(1)$. They get $\taux=O(1)$ when the observation pattern is uniformly distributed over all patterns with $\nid=N/C$ and $\ndj=N/R$. Both $R$ and $C$ much grow linearly with $N$.  They
then require $\nid=O(1)$ and $\ndj=O(1)$.

Our balance criteria are minimal.  Our main results require $\max(\max_i\nid,\max_j\ndj)/N=o(1)$. Both $R$ and $C$
can grow with $N$ at different rates. Our simulations, but not our theory, use a balance condition with nearly equal values among the $\nid$ and nearly equal values among the $\ndj$.

We study the probit model because its Gaussian latent variable is a good match for Gaussian random effects.
The probit and logit link functions are nearly
proportional outside of tail regions \citep[pp.~246--247]{agresti:2002},
and so they often give similar results.

Crossed models also differ from hierarchical models in that it is only recently that the maximum likelihood estimate for generalized linear mixed models has been shown to be consistent for crossed random effects.
This was accomplished by the subset argument of \cite{jiang:2013}
showing that the score equations have a root that is consistent
for the parameter.
One of the main open questions in generalized linear
mixed models, like the ones we consider here, is at what rate do the estimators of the model parameters converge?  Intuitively one might
expect $O(N^{-1/2})$ or $O\{\min(R,C)^{-1/2}\}$, in accordance with the results  valid for nested designs 
\citep{jiang2022usable}.
In some settings estimators for different parameters converge at
different rates.  See \cite{jiang:2013} for some discussion and \cite{lyu:siss:wels:2024} for recent work assuming balanced sampling. %

\section{The All-Row-Column Method}\label{sect:ARC}

\subsection{Preliminaries}\label{sect:preliminaries}

We consider two crossed random effects.
There is a vector ${a}\in\Real^R$ with elements $a_i$
and another vector ${b}\in\Real^C$ with elements $b_j$.
Conditionally on ${a}$ and ${b}$, the $Y_{ij}$ are independent with
\begin{equation}\label{eq:ygivab}
\pr( Y_{ij}=1\giv {a},{b} ) = \Phi({x}_{ij}\T{\beta}+a_i+b_j),
\end{equation}
where $\Phi(\cdot)$ is the standard normal cumulative distribution function. In these models, random effects are typically assumed to be uncorrelated,
${a}\sim\dnorm(0,\ssa I_R)$
independently of ${b}\sim\dnorm(0,\ssb I_C),$ where $I_n$ is the $n \times n$ identity matrix. See, for example, \citeauthor{McCullagh:Nelder:1989} (\citeyear{McCullagh:Nelder:1989}, page 444).
The probit model has a representation in
terms of latent variables $\err_{ij}\sim \dnorm(0,1)$, via
\begin{equation}\label{eq:ygivablatent}
Y_{ij}=1\{{x}_{ij}\T{\beta}+a_i+b_j+\err_{ij}>0\},
\end{equation}
where $1\{E\}$ is the indicator function of the event $E$.
The probability~\eqref{eq:ygivab} and the likelihoods derived from it are all conditional on the values of ${x}_{ij}$. 

We write $\obset\subset \{1,\dots,R\}\times\{1,\dots,C\}$
for the set of $(i,j)$ pairs where $({x}_{ij},Y_{ij})$ was observed. 
We also work conditionally on $\obset$.
In our motivating applications the pattern of observation/missingness
could be informative.  Addressing that issue would necessarily require information from
outside the data.  Furthermore, the scaling problem is still the subject of current research even in the
non-informative missingness setting. Therefore we consider estimation strategies without taking account of missingness.

The likelihood for ${\theta}=({\beta}\T, \ssa, \ssb)\T$ is a cumbersome integral of size $R+C$,
\begin{equation}\label{eq:fulllik}
L({\theta}) = \siga^{-R}\sigb^{-C} \int_{\Real^{R+C}}  L({\beta} \giv {a}, {b}) \prod_{i=1}^R  \varphi\Bigl(\frac{a_i}{\siga}\Bigr)  \prod_{j=1}^C   \varphi\Bigl(\frac{b_j}{\sigb} \Bigr) \rd {a} \rd {b},
\end{equation}
where $\varphi(\cdot)$ is the standard normal probability density function, and $L({\beta} \giv {a}, {b})$ is the conditional likelihood of ${\beta}$ given the random effects.
The conditional likelihood we need is
$$
L({\beta} \giv {a}, {b})=\prod_{(i,j) \in \obset} \Phi({x}_{ij}\T {\beta} + a_i + b_j)^{y_{ij}}  \Phi(-{x}_{ij}\T {\beta} - a_i - b_j)^{1-y_{ij}},
$$
and $L(\theta)$ is commonly called the marginal likelihood. 

The first-order Laplace approximation is a standard approach to approximate the integral in the marginal likelihood. The Laplace algorithm maximizes the logarithm of the integrand in the marginal likelihood \eqref{eq:fulllik} over ${a}\in\Real^R$ and ${b}\in\Real^C$ for fixed ${\theta}$. It then multiplies the maximum value of the integrand by $\det\{H^{-1/2}(\theta)\}$ where $H(\theta)\in\Real^{(R+C)\times(R+C)}$ is the Hessian of the log integrand with respect to ${a}$ and ${b}$ for fixed ${\theta}$. The result is an approximate marginal likelihood $\tilde L({\theta})$ that is optimized to get $\hat{{\theta}}$.
 See for example \cite{ogden2021error} or \cite{shun:mccu:1995}. 
 If the square root of the inverse Hessian is computed by standard methods then that alone has cost of $\Omega(N^{3/2})$ by the argument for generalised linear models discussed in \S\ref{sect:intro}.  Similarly, if the inner optimization is done using Newton steps, that will have a cost that is $\Omega(N^{3/2})$ per iteration. We return to this issue in \S\ref{sec:compcost} where a Laplace approximation shows superlinear cost that is $o(N^{3/2})$. 

Even if the Laplace approximation were computable for large $N$, it does not provide asymptotically valid results in our context because the size of the likelihood integral corresponds to the number of random effects $\Omega(N^{1/2})$ and therefore grows too fast with the sample size to ensure consistent results. See \cite{shun:mccu:1995}, \cite{ogden2021error} and \cite{Tang:2024} for a discussion of the conditions that must be satisfied to make the Laplace approximation reliable when the size of the integral grows with the sample size.

\subsection{Scalable composite likelihood inference}\label{sect:scalable_CL}
Our approach for a consistent and scalable estimator in high-dimensional probit models with crossed-random effects combines estimates for three misspecified probit models. Each of them is constructed through the omission of some random effects.  By combining~\eqref{eq:ygivab} and~\eqref{eq:ygivablatent} we
find that marginally
\begin{equation}\label{eq:noanob}
\pr( \yij = 1 ) = \Phi( {x}_{ij}\T{\gamma}),
\end{equation}
for ${\gamma} = {\beta}/(1+\ssa+\ssb)^{1/2}$.
The proposed method begins with estimation of  ${\gamma}$ from the na\"ive model (\ref{eq:noanob}) that omits both of the random effects through maximization of the likelihood 
\begin{align}\label{eq:lall}
L_{\text{all}}({\gamma})= \prod_{(i,j) \in \mathcal S} \Phi({x}_{ij}\T {\gamma})^{y_{ij}}\Phi(-{x}_{ij}\T {\gamma})^{1-y_{ij}}.
\end{align}
Maximization of this likelihood 
requires neither high dimensional integrals nor expensive algebra and we see it taking $O(N)$ computation in our examples. 
We then need estimates of
$\ssa$ and $\ssb$ to get the scale right. Our analysis of~\eqref{eq:noanob} will also account for within-row and within-column correlations among the $\yij$.
While $\sign({\gamma}_k)=\sign(\beta_k)$  $(k=1,\dots,p),$
confidence intervals for ${\gamma}_k$ based on model~\eqref{eq:noanob}
would be na\"ive if they did not account for the
dependence among the responses. 

Consider the reparameterization ${\psi}=({\gamma}\T, \tau^2_A, \tau^2_B)\T,$ where
\begin{equation}\label{eq:reparam}
{\gamma}=\frac{{\beta}}{(1+\ssa+\ssb)^{1/2}}, \quad \tau^2_A=\frac{\ssa}{1 + \ssb}, \quad \tau^2_B=\frac{\ssb}{1 + \ssa}.
\end{equation}
After dividing ${x}_{ij}\T{\beta}+a_i+b_j+\err_{ij}$ by $(1+\ssa)^{1/2}$ or by $(1+\ssb)^{1/2}$, model~\eqref{eq:ygivablatent} implies that
\begin{align}
\pr( \yij = 1\giv{a} ) &= \Phi( {x}_{ij}\T{\gamma}_A+u_i), 
\label{eq:nob}\\
\pr( \yij = 1\giv{b} ) &= \Phi( {x}_{ij}\T{\gamma}_B+v_j), \label{eq:noa}
\end{align}
where ${u}\sim\dnorm(0,\tau^2_AI_R)$ and
${v}\sim\dnorm(0,\tau^2_BI_C)$ for  ${\gamma}_A={\gamma}(1+\tau^2_A)^{1/2}$
and ${\gamma}_B={\gamma}(1+\tau^2_B)^{1/2}$. Given $\hat{{\gamma}}$, the maximizer of~\eqref{eq:lall}, we proceed with estimation of $\tau^2_A$ and $\tau^2_B$ from the two probit models (\ref{eq:nob}) and (\ref{eq:noa}) that each omit one of the random effects. Fitting 
models~\eqref{eq:nob} and~\eqref{eq:noa} involve simpler integrals than~\eqref{eq:ygivab} since the  latent variable representations of these models
\begin{equation*}
 Y_{ij}=1\{{x}_{ij}\T{\gamma}_A+u_i+\err_{ij}>0\}, \quad
 Y_{ij}=1\{{x}_{ij}\T{\gamma}_B+v_j+\err_{ij}>0\},
\end{equation*}
have hierarchical (not crossed) error structures. This is where we are able to replace the $(R+C)
$-dimensional integral \eqref{eq:fulllik} by $R+C$ univariate ones. 
Model~\eqref{eq:nob} is fitted by maximizing the row-wise likelihood
\begin{equation}\label{eq:rowlik}
L_{\text{row}}(\tau^2_A)=\tau_A^{-R}\prod_{i=1}^R \int_{\Real} L_{i\sumdot}(\hat{{\gamma}}_A  \giv u_i)  \varphi\Bigl(\frac{u_i}{\tau_A}\Bigr)  \rd u_i,
\end{equation}
where $L_{i\sumdot}(\hat{{\gamma}}_A \giv u_i)$ is the conditional likelihood of $\hat{{\gamma}}_A=\hat{{\gamma}}(1+\tau^2_A)^{1/2}$ given $u_i,$
\begin{align}\label{eq:rowilikelihood}
L_{i\sumdot}(\hat{{\gamma}}_A  \giv u_i)=\prod_{j \giv i}  \Phi({x}_{ij}\T  \hat{{\gamma}}_A  +u_i)^{y_{ij}}    \Phi(-{x}_{ij}\T  \hat{{\gamma}}_A  -u_i)^{1-y_{ij}},
\end{align}
where $j\giv i={\{j: (i,j) \in \mathcal S\}}$ is the set of indices $j$ such that $({x}_{ij},Y_{ij})$ is observed. The row-wise likelihood~\eqref{eq:rowlik} is a product of $R$ one-dimensional integrals and  
is a function of $\tau^2_A$ only, because we fix ${\gamma}$ at the estimate $\hat{{\gamma}}$ obtained from the maximization of the all likelihood at the previous step. Rows with a single observation do not contribute to estimation of $\tau^2_A$ because 
$\pr(Y_{ij}=1)  = \Phi({x}_{ij}\T  \hat{{\gamma}}),$ which does not depend on $\tau^2_A.$ 
Reversing the roles of the rows and columns, we get a column-wise likelihood $L_{\text{col}}$ which we maximize to get an estimate $\hat\tau^2_B$ of $\tau^2_B$.

Finally, 
we invert the equations in~\eqref{eq:reparam} to get
$$
\hat{{\beta}} = \hat{{\gamma}} (1 + \hat\sigma_A^2 + \hat \sigma_B^2)^{1/2}, \quad \hat \sigma_A^2=\frac{\hat \tau_A^2(1 + \hat \tau^2_B)}{1 -  \hat \tau^2_A \hat \tau^2_B}, \quad
 \hat \sigma_B^2=\frac{\hat \tau_B^2(1 + \hat \tau^2_A)}{1 -  \hat \tau^2_A \hat \tau^2_B}.
$$
Note that by the definitions in \eqref{eq:reparam}, $\tau^2_A\tau^2_B<1$.
In our computations, we never encountered a setting where $\hat\tau^2_A\hat\tau^2_B\geq 1$ and so our estimates  $\hat \sigma^2_A$ and $\hat\sigma^2_B$ were never negative.

We call our method `all-row-column'.
The name comes from model~\eqref{eq:noanob} that uses
all the data at once, model~\eqref{eq:nob} that combines likelihood
contributions from within each row and model~\eqref{eq:noa} that
combines likelihood contributions from within each column. Figure \ref{fig:ARC} illustrates those models for a data set of $N=39$ observations in $R=10$ rows and $C=10$ columns. In each of three misspecified probit models, points in different boxes are assumed to be independent. The left panel of Figure \ref{fig:ARC} illustrates the `all' model with independent data. The next two illustrate the `row' and `column' hierarchical models. We combine fits from these three misspecified models to get consistent scalable formulas despite the dependencies involved. This approach is a new form of composite likelihood \citep{lindsay:1988,varin:2011}. An earlier version of composite likelihood that was applied to crossed random effects is discussed in \cite{bellio:2005}. They considered a standard pairwise likelihood which is however not scalable and thus inappropriate for our problem. 
Instead, our all-row-column method has $O(N)$ cost per iteration and we  get $O(N)$ cost empirically.

\begin{figure} 
\centering
\includegraphics[width=.9\hsize]{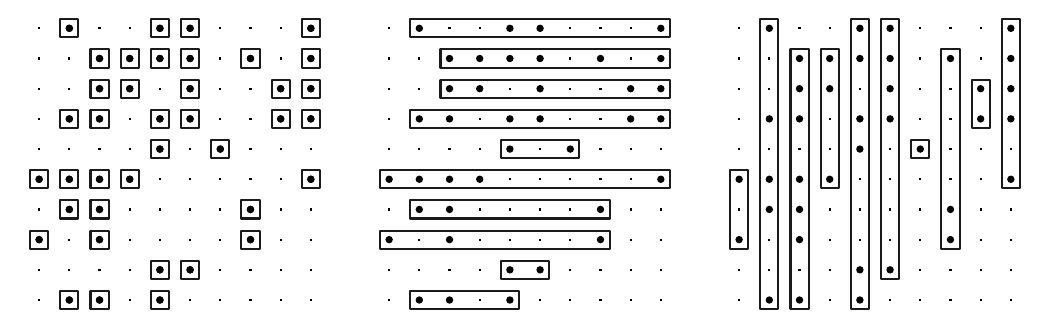}
\caption{\label{fig:ARC} 
The all-row-column method for $N=39$ observations (marked $\bullet$) in $R=10$ rows and $C=10$ columns.
}
\end{figure}

Computation of the row- and column-wise likelihoods require one to approximate up to $R$ and $C$ univariate integrals of a standard hierarchical probit model. Since $R$ and $C$ are large in the applications that motivate this work, then accurate approximation of the one-dimensional integrals is crucial. Otherwise, the accumulation of approximation errors could induce serious biases in the estimation of $\tau^2_A$ and $\tau^2_B$. Our intensive numerical studies indicated that accurate approximation of the integrals is obtained by the well-established adaptive Gauss-Hermite quadrature \citep[\emph{e.g.},][]{liu1994note}  with a suitable choice of the number of quadrature nodes. 

In the next Theorem, we prove the weak consistency of the maximizer $\hat{{\gamma}}$ of the `all' likelihood~\eqref{eq:lall} to the value ${\gamma}$ of equation~\eqref{eq:reparam}.  First we introduce some notation.
Let $\zij=1$ if $({x}_{ij},\yij)$ is observed
and take $\zij=0$ otherwise.
The number of observations in row $i$ is $\nid = \sum_{j=1}^C\zij$ and similarly column $j$ has $\ndj=\sum_{i=1}^R\zij$ observations in it. Let $\epsilon_R=\max_i\nid/N$ and
$\epsilon_C =\max_j\ndj/N$.  
We assume
that $\max(\epsilon_R,\epsilon_C)\to0$
as $N\to\infty$.
We take $\zij$ to be deterministic 
with at most one
observation for any $(i,j)$ pair.  In our motivating applications, we would seldom if ever have multiple observations for one $(i,j)$ pair.  Even then, one might only use the most recent of those observations.  The ${x}_{ij}$
are not dependent on the $\zij$.
Finally, $\yij$ are sampled from their probit distribution
conditionally on~${x}_{ij}$.

\begin{theorem}
Let $\yij\in\{0,1\}$ follow the crossed random effects probit model~\eqref{eq:ygivab}  
with true value ${\gamma}_0$ for the parameter 
${\gamma}={\beta}/(1+\ssa+\ssb)^{1/2}.$ 
Let the number of observations $N\to\infty$
while $\max(\epsilon_R,\epsilon_C)\to0$.
Let ${x}_{ij}\in\Real^p$ satisfy
\begin{enumerate}
\item $\Vert{x}_{ij}\Vert \le B<\infty;$
\item $N^{-1}\sum_{ij}\zij{x}_{ij}{x}_{ij}\T\to V\in\Real^{p\times p}, $ where $V$ is positive definite;
 \item there is no nonzero vector ${v} \in \Real^p$ such that ${v}\T {x}_{ij}\geq 0$ for all $(i, j)$ with $Z_{ij}=1$ and $y_{ij}=1$ and  ${v}\T {x}_{ij}\leq 0$ for all $(i, j)$ with $Z_{ij}=1$ and $y_{ij}=0$.
\end{enumerate}
Let $\hat{{\gamma}}\in\Real^p$ be any maximizer of~\eqref{eq:lall}.
Then for any $\epsilon>0$
$$
\pr( \Vert\hat{{\gamma}}-{\gamma}_0\Vert>\epsilon)\to0, \;
\text{as $N\to\infty$}.
$$
\end{theorem}
The proof  is reported in Section S1 of the Supplementary Material and we obtained it by adapting the proof strategy of \cite{luml:hamb:2003} to our setting. The balance condition that no single row or column have a non-vanishing fraction of data is much weaker than the norm for the crossed random effects literature.

Now we consider consistent estimation of $\tau^2_A$ and $\tau^2_B$ from the row and column likelihoods, respectively.  Those provide consistent estimates of $\ssa$ and $\ssb$ with which one can then adjust the consistent estimate of ${\gamma}$ to get a consistent estimate of ${\beta}$. 

\begin{theorem}\label{thm:cramer}
 Under the assumptions of Theorem 1, there is a root of the row likelihood equation that is a consistent estimator for $\tau^2_A$ and a root of the column likelihood equation that is a consistent estimator for $\tau^2_B.$
 \end{theorem}

Our proof of this theorem is also reported in \S S1 of the Supplementary Material. The proof uses the subset argument of \cite{jiang:2013} to show Cramer consistency of the maximum row likelihood  estimator of $\tau^2_A$
and equivalently of the maximum column likelihood estimator of $\tau^2_B$.
  
\subsection{Robust sandwich variance}\label{sect:standard_errors}

After a customary Taylor approximation, the  variance of $\hat{{\theta}}$ is   $ \var(\hat{{\theta}})  \doteq D\,  \var(\hat{{\psi}})  \, D\T,$ with $D$ the Jacobian matrix of the reparameterization from ${\theta}$ to ${\psi}$.  Letting  $\psi_0$ denote the true value of $\psi$ and 
${u}_{\text{arc}}({\psi})$ denote the score vector of the all-row-column estimator  constructed by stacking the scores of the three misspecified likelihoods, the asymptotic variance of $\hat{{\psi}}$ is
\begin{equation}\label{eq:sandwich}
\avar(\hat{{\psi}}) =\mathcal{I}_{\text{arc}}^{-1}({\psi}_0) V_{\text{arc}}({\psi}_0) \mathcal{I}_{\text{arc}}^{-1}({\psi}_0),
\end{equation}
where $\mathcal{I}_{\text{arc}}(\psi)=-\E\{\partial {u}_{\text{arc}}(\psi)/ \partial \psi \}$ and $ V_{\text{arc}}(\psi)=\var\{u_{\text{arc}}(\psi)\}$ are the expected information and the score variance for the all-row-column estimator. These two matrices are not equal because the second Bartlett identity does not hold for the misspecified likelihoods that constitute the all-row-column method. While estimation of the `bread' matrix $\mathcal{I}_{\text{arc}}$ of the sandwich is not problematic, direct computation of the `filling' matrix $ V_{\text{arc}}$ is not feasible in our large-scale setup because it requires us to approximate a large number of multidimensional integrals with a cost that does not meet our $O(N)$ constraint.

Since all-row-column estimates require $O(N)$ computations, we can estimate the variance of $\hat{\theta}$ with a parametric bootstrap. However, it is preferable to evaluate the estimation uncertainty without assuming the correctness of the fitted model, for example using the nonparametric pigeonhole bootstrap described in \cite{owen:2007}.
In that approach, the rows in the data set are resampled
independently of the columns.  So if a row is included twice
and a column is included three times, the corresponding element
is included six times.  The resulting bootstrap variance
for a mean (such as one in a score equation)
over-estimates the random effects variance
by an asymptotically negligible amount.  It does not require
homoscedasticity of either the random effects or the errors.

Now we describe a convenient approach that we have developed to estimate the variance of 
$\hat{{\beta}}$ in $O(N)$ operations without the need of resampling and repeated fitting as in the bootstraps mentioned above. The partitioned expected all-row-column information matrix is
\begin{align*}
\mathcal{I}_{\text{arc}}({\psi}) &=-
\begin{pmatrix}
\E(\partial^2 \ell_{\text{all}}  / \partial {\gamma} \partial {\gamma}\T)  & 0 & 0 \\[0.5em]
\E(\partial^2 \ell_{\text{row}}  / \partial {\gamma} \partial \tau^2_A) & \E(\partial^2  \ell_{\text{row}}  / \partial \tau^2_A \partial \tau^2_A) & 0 \\[0.5em]
\E(\partial^2 \ell_{\text{col}} / \partial {\gamma} \partial \tau^2_B)  & 0 & \E(\partial^2 \ell_{\text{col}}  / \partial \tau^2_B \partial \tau^2_B) \\
\end{pmatrix},
\end{align*}
where $\ell_{\text{all}}=\log L_{\text{all}}(\gamma)$ from~\eqref{eq:lall}, $\ell_{\text{row}}=\log L_{\text{row}}(\tau^2_A)$ from~\eqref{eq:rowlik} and $\ell_{\text{col}}=\log L_{\text{col}}(\tau^2_B)$. 
Since $\mathcal{I}_{\text{arc}}({\psi})$ %
is triangular, then the asymptotic variance for 
$\hat{{\gamma}}$ is
\begin{equation}\label{eq:sandwich_gamma}
\avar(\hat{{\gamma}})=\mathcal{I}_{\text{all}}^{-1}({\gamma}_0) V_{\text{all}}({\gamma}_0) \mathcal{I}_{\text{all}}^{-1}({\gamma}_0),
\end{equation}
where $\mathcal{I}_{\text{all}}({{\gamma}})=-\E(\partial^2 \ell_{\text{all}}  / \partial {\gamma} \partial {\gamma}\T)$  and $V_{\text{all}}({{\gamma}})=\var(\partial \ell_{\text{all}}  / \partial {\gamma})$ are the Fisher expected information and the score variance of the all likelihood.  The asymptotic variance (\ref{eq:sandwich_gamma}) for $\hat{{\gamma}}$ is thus the same as that of the estimator that maximizes the all likelihood when the nuisance parameters $\tau^2_A$ and $\tau^2_B$ are known. The robust sandwich estimator of the variance of $\hat{{\gamma}}$ is obtained by replacing $\mathcal{I}_{\text{all}}({\gamma}_0)$ and $V_{\text{all}}({\gamma}_0)$ with some estimators that are consistent and robust to misspecification. The expected information is naturally estimated with the observed information,
\begin{align*}
J_{\text{all}}(\hat{{\gamma}})=\sum_{ij}Z_{ij}\varphi(\hat\eta_{ij})
\biggl\{y_{ij} \frac{\varphi(\hat\eta_{ij})+\hat\eta_{ij}\,\Phi(\hat\eta_{ij})}{\Phi(\hat\eta_{ij})^2}
+ 
(1-y_{ij}) \frac{\varphi(\hat\eta_{ij})-\hat\eta_{ij}\,\Phi(-\hat\eta_{ij})}{\Phi(-\hat\eta_{ij})^2}
\biggr\} {x}_{ij}{x}_{ij}\T,
\end{align*}
where $\hat\eta_{ij}={x}_{ij}\T \hat{{\gamma}}.$  Estimation of $V_{\text{all}}({\gamma}_0)$ is more involved. This matrix can be decomposed into the sum of three terms, 
\begin{align*}
V_{\text{all}}({\gamma}_0)&=\sum_{ij} \sum_{rs} Z_{ij} Z_{rs} \E\bigl\{{u}_{ij}({\gamma}_0) {u}_{rs}\T({\gamma}_0)\bigr\} \\
&=\sum_{ijs} Z_{ij} Z_{is} \E\bigl\{{u}_{ij}({\gamma}_0) {u}_{is}\T({\gamma}_0)\bigr\} +
\sum_{ijr} Z_{ij} Z_{rj} \E\bigl\{{u}_{ij}({\gamma}_0) {u}_{rj}\T({\gamma}_0)\bigr\}  \\&-  
\sum_{ij} Z_{ij}  \E\bigl\{{u}_{ij}({\gamma}_0) {u}_{ij}\T({\gamma}_0)\bigr\},
\end{align*}
where ${u}_{ij}({\gamma})$ is the score for a single observation $Y_{ij},$
$$
{u}_{ij}({\gamma})=\frac{\varphi({x}_{ij}\T {\gamma}) \bigl\{y_{ij}-\Phi({x}_{ij}\T {\gamma}) \bigr\} {x}_{ij}}{\Phi({x}_{ij}\T {\gamma})\Phi(-{x}_{ij}\T {\gamma})}.
$$
The corresponding estimator of $V_{\text{all}}({\gamma}_0)$ is
$$
\wh{V}_{\text{all}}(\hat{{\gamma}})=\wh{V}_A(\hat{{\gamma}})+\wh{V}_B(\hat{{\gamma}})-\wh{V}_{A \cap B}(\hat{{\gamma}}),
$$
whose components are computed grouping the individual scores with respect to each random effect and their interaction,
$$
\wh{V}_A(\hat{{\gamma}})= \sum_i {u}_{i \sumdot}(\hat{{\gamma}})    {u}_{i\sumdot}\T(\hat{{\gamma}}),\quad
\wh{V}_B(\hat{{\gamma}})= \sum_j  {u}_{\sumdot j}(\hat{{\gamma}})  {u}_{\sumdot j}\T(\hat{{\gamma}}), \quad
\wh{V}_{A \cap B}(\hat{{\gamma}}) =   \sum_{i j} Z_{ij} {u}_{ij}(\hat{{\gamma}})  {u}_{ij}\T(\hat{{\gamma}}),
$$
where ${u}_{i\sumdot}({\gamma})=\sum_{j}  Z_{ij} {u}_{ij}({\gamma})$ and ${u}_{\sumdot j}({\gamma})=\sum_{i} Z_{ij}  {u}_{ij}({\gamma}).$ 
Estimators of the form $\wh{V}_{\text{all}}(\hat{{\gamma}})$ are used in statistical modeling of data clustered within multiple levels in medical applications \citep{miglioretti:2004} and in economics \citep{cameron:2011} where they are known as two-way cluster-robust sandwich estimators.

Finally, we approximate the variance of $\hat{{\beta}}$ by plugging in the estimates $\wh\sigma_A^2$ and $\wh\sigma_B^2$,
\begin{equation}
\wh\var
(\hat{{\beta}}) =  (1+\wh\sigma_A^2+\wh \sigma_B^2) \, \wh{\var}(\hat{{\gamma}}) = (1+\wh\sigma_A^2+\wh \sigma_B^2)\, {J}_{\text{all}}^{-1}(\hat{{\gamma}})\wh{V}_{\text{all}}(\hat{{\gamma}}) {J}_{\text{all}}^{-1}(\hat{{\gamma}}). \label{eq:var-beta}
\end{equation}

A limitation of this approach is that it neglects the uncertainty in the estimation of the variance components: Although 
we do not expect
a substantial impact in high-dimensional applications, it could be possible to adjust the (\ref{eq:var-beta}) for the variability of the variance components through bootstrapping the row-wise and column-wise estimates $\hat\tau^2_A$ and $\hat\tau^2_B$.

\section{Simulations}\label{sect:simu}

In this section, we simulate from the probit model with crossed random effects (\ref{eq:ygivab}) and compare the performance of the all-row-column estimator with the traditional estimator obtained by maximizing the first-order Laplace approximation of the likelihood. A further method considered is an `infeasible oracle estimator'  that uses the unknown true values of $\ssa$ and $\ssb$ to estimate the regression parameters as $\hat{{\beta}}_{\mathrm{oracle}} = \hat{{\gamma}}(1+\ssa+\ssb)^{1/2}.$ The all-row-column method instead corrects $\hat{{\gamma}}$ using estimates of the variance components. All methods were implemented in the R language. The package \texttt{TMB}  \citep{
kristensen:2016} was used for the Laplace approximation, with the \texttt{nlminb} optimization function employed for its maximization. We did not use the popular \texttt{glmer} function from the {R} package \texttt{lme4} \citep{Bates:2015} because the current version of \texttt{TMB} is substantially more computationally efficient and thus allowed us to compare our method to the first-order Laplace approximation at larger dimensions than it would otherwise be  possible with \texttt{glmer}. 
The row- and column-wise likelihoods were coded in C++ and integrated in {R} with \texttt{Rcpp} \citep{Rcpp}, and optimized by Brent's method as implemented in the \texttt{optimise}
function.  

\subsection{Simulation settings}
We considered eight different settings defined by 
combining three binary factors.  The first factor is 
whether the simulation
is balanced (equal numbers of rows and columns)
or imbalanced (with very unequal numbers) 
like we typically see in applications.
The second factor is whether the regression model is null
apart from a nonzero intercept or has nonzero regression
coefficients.  The third factor is whether the
random effect variances are set at a high level
or at a low level. 
Given $R$ and $C$, the set $\obset$ is obtained by independent and identically distributed
Bernoulli sampling with probability $1.27\times N/(RC)$. The value $1.27$ is the largest one for which \cite{ghosh:2020} could prove that backfitting takes $O(1)$ iterations.
This sampling makes
the attained value of $N$ random but with a very
small coefficient of variation.

We denote by $R=N^\rho$ and $C=N^\kappa$ the number of rows and columns in the data, as a power of the total sample size $N$. The two levels of
the balance factor are thereafter termed `balanced' with   $\rho=\kappa=0.56$ and
 `imbalanced' with  $\rho=0.88$ and $\kappa=0.53$.
 The balanced case was used in \cite{ghosh:2021}, and the imbalanced case is similar to the Stitch Fix data in \S \ref{sect:application}.
  Because $\rho+\kappa>1$, the fraction of possible observations
  in the data is $N/(RC)=N^{1-\rho-\kappa}\to 0$ as $N\to\infty$
  providing asymptotic sparsity in both cases.
  While the first choice has $R/C$ constant,
  the second choice has $R/C\to\infty$ with $N$.
  We believe that this asymptote is a better description
  of our motivating problems than either a setting with
  $C$ fixed as $R\to\infty$ or the setting common
  in random matrix theory \citep{edel:rao:2005} where $R$ and $C$ diverge
  with $R/C$ approaching a constant value. Notice that when $\rho=0.88$ and $\kappa=0.53$
  the condition $\max(\rho+2\kappa,2\rho+\kappa)<2$ invoked for the estimator in \cite{ghosh:2021}
  does not hold.

We considered seven predictors generated from a multivariate zero-mean normal distribution with a covariance matrix $\Sigma$ corresponding to an autocorrelation process of order one, so that the entry $(k, l)$ of $\Sigma$ is $\phi^{|k-l|}$. We set $\phi=0.5$ in all the simulations. We always used the intercept $\beta_{0} = -1.2$ because in our applications $\pr( Y=1)<1/2$ is typical. For the predictor coefficients we considered two choices termed `null' with $\beta_{\ell} = 0$ and `linear' with  $\beta_{\ell} = -1.2 + 0.3 \ell$ $(\ell=1,\dots,7).$
The first setting is a null one where ${x}$ is not predictive at all, while the second setting has modestly important nonzero predictors
whose values are in linear progression. 

The  two choices for the variance component parameters are termed `high variance' with $\siga = 1$ and $\sigb = 1$ and `low variance' with $\siga = 0.5$ and $\sigb = 0.2$.
We chose the first setting to include variances higher than what is typically observed in applications. The second setting is closer to what we have seen in data such as that in  \S \ref{sect:application}. 
We represent the eight settings with
mnemonics as shown in Table~\ref{tab:settings}.
For example, \Stoo\  means row-column 
imbalance ($\rho=0.88, \kappa=0.53$), all
predictor coefficients are zero 
and the main effect variances are large ($\siga=1, \sigb=1$). 

\begin{table}
\centering
{
\begin{tabular}{lcclrr}
{Setting} & \multicolumn{2}{c}{{Sparsity}} & {Predictors} & \multicolumn{2}{c}{{Variances}} \\
& $\rho$ & $\kappa$ & & $\siga$ & $\sigb$ \\
\Sooo & 0.56 & 0.56 & all zero & 1.0 & 1.0\\
\Stoo & 0.88 & 0.53 & all zero &  1.0 & 1.0\\
\Soto & 0.56 & 0.56 & not all zero &  1.0 & 1.0\\
\Stto & 0.88 & 0.53  & not all zero &  1.0 & 1.0\\
\Soot & 0.56 & 0.56 & all zero &  0.5 & 0.2 \\
\Stot
& 0.88 & 0.53 & all zero &  0.5 & 0.2 \\
\Sott
& 0.56 & 0.56 & not all zero &  0.5 & 0.2 \\
\Sttt
& 0.88 & 0.53 & not all zero &  0.5 & 0.2 
\end{tabular}}
\label{tab:settings}
\caption{Summary of the eight simulation settings.}
\end{table}

For each of these eight settings, we considered 13 increasing sample sizes $N$ in the interval from $10^3$ to $10^6$ obtained by taking 13 equispaced values on the $\log_{10}$ scale. 
As we see next, the Laplace method had
a cost that grew superlinearly and to keep
costs reasonable we only used sample sizes
up to $10^5$ for that method.
For each of these 13 sample sizes and each of the eight settings, we simulated 1000 data sets.

As suggested by a referee, we experimented with different values for the number of quadrature nodes $k$ to approximate the univariate integrals of the row- and column-likelihoods; see also the recent work
\cite{Bilodeau:2024}. 
The value $k=1,$ which corresponds to the Laplace approximation, produced estimates of $\sigma_A$ and $\sigma_B$ affected by substantial downward bias even at large $N$. The bias disappeared by increasing $k$. With $k=5,$ we obtained results not affected by bias and essentially indistinguishable from those with $k=25$. Therefore, the results discussed in the rest of this section were calculated with $k=5$ nodes.

Graphs comparing the computational costs, the statistical properties and the scalability of the three estimation methods for all  eight settings are reported in the Supplementary Materials. 
To save space, we present graphs for only one of the settings, \Stoo, and just summarize the other settings. 
This chosen setting is a challenging one.
It is not surprising that imbalance and
large variances are challenging.  
The binary regression setting
is different from linear modeling where
estimation difficulty is unrelated to
predictor coefficient values. 
The main reason to highlight this setting is
that it illustrates an especially bad
outcome for the Laplace method estimate
of $\siga$.  Similar but less extreme
difficulties for the Laplace method's
estimates of $\siga$ appear in
setting \Stot.

\subsection{Computational cost}\label{sec:compcost}
Figure \ref{fig:times211} shows the average computation times for the three methods, obtained
on a 16-core 3.5 GHz AMD processor equipped with 128 GB of RAM. 
It also shows regression lines of
log cost versus $\log N$, marked
with the regression slopes.
The all-row-column and oracle method's slope are both very nearly $1$ as
expected.
The Laplace method slope is clearly
larger than $1$ and, as noted above, we curtailed the
sample sizes used for that case to
keep costs reasonable. If we extrapolate the Laplace cost to $N=10^8$, comparable to the Netflix data \citep{bennett2007netflix}, then the cost grows past 12.9 days while the all-row-column cost grows only to about 45 minutes. 
The computational cost of all-row-column can be further reduced with parallel computing by distributing the approximations of the R+C one-dimensional integrals across multiple cores.

\begin{figure}[!h]
  \centering
  \includegraphics[width=.9\hsize]{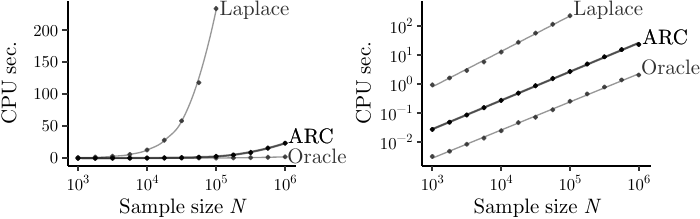}
 \caption{The left panel shows 
    time in seconds versus sample size $N$ for
    the \Stoo\ setting for all three methods:
    Laplace, all-row-column and the Oracle. The right panel is the same plot with times displayed on the $\log_{10}$ scale. 
   }
    \label{fig:times211}
\end{figure}

Similar estimated computational costs  were obtained for the other seven settings as %
summarized in Table~\ref{tab:timings} which shows estimated
computational cost rates for all eight
very close to $1$.
The oracle slope is 
consistently close to $1$ in imbalanced
settings and somewhat bigger than $1$ in balanced ones.
The Laplace slope is
consistently larger than $1$. 
It tends to be higher for balanced
settings although one of the imbalanced cases
also has a large slope.

\begin{table}
\centering
\begin{tabular}{lccc}
{Setting} &  {Oracle slope} &  {ARC slope} & {Laplace slope}\\
\Sooo & 1.06 & 1.01 & 1.30\\
\Soot & 1.05 & 1.01 & 1.24\\
\Soto & 1.06 & 1.00 & 1.29\\
\Sott & 1.07 & 0.99 & 1.26\\
\Stoo & 0.98 & 1.00 & 1.23\\
\Stot & 0.98 & 1.00 & 1.13\\
\Stto & 1.02 & 0.99 & 1.15\\
\Sttt & 0.98 & 1.00 & 1.15\\
\end{tabular}
\caption{
Computational cost for all eight settings.
Linear regression slopes for $\log(\text{CPU seconds})$ versus $\log(N).$
}
\label{tab:timings}
\end{table}

We have investigated the data behind the oracle slopes, but cannot yet explain the mild superlinearity we sometimes see.  The number of
 Fisher-scoring iterations used by the oracle method varies at small sample sizes but is consistently near seven at larger sample sizes where the superlinearity is more prominent.  We have seen some outliers in the computation times at small sample sizes but not at large sample sizes.  Replacing the means at different $N$ by medians does not materially change the slopes. We consider the amount of unexplained nonlinearity to be small but not negligible.  For example, when going from $N=10^3$ to $N=10^6$ a rate like $N^{1.06}$ yields a roughly 1500 fold cost increase instead of the expected 1000 fold increase. The anomaly is concentrated in the balanced simulations, but knowing this has not been enough to identify the cause.

\subsection{Regression coefficient estimation}

Next, we turn to estimation of the regression coefficients, treating the intercept differently from the others.
The intercept poses a challenge because it is somewhat confounded with the random effects.
For instance, if we replace
$\beta_0$ by $\beta_0+\lambda$ while
replacing $\ai$ by $\ai-\lambda$
then the $\yij$ are unchanged. Large $\lambda$
would change $\bar a = (1/R)\sum_{i=1}^R\ai$
by an implausible amount that should be
statistically detectable given
that $\ai\sim\dnorm(0,\ssa)$.
On the other hand $|\lambda|=O(\siga R^{-1/2})$ would be hard to detect
statistically.
The other regression parameters are
not similarly confounded with main
effects in our settings.
\cite{ghosh:2021} saw that a categorical predictor that is a function of just the row index $i$ or just the column index $j$ brings a similar confounding.

Because of the confounding described
above, we anticipate that the true mean square error 
rate for the intercept cannot be better than
$O\{\min(R,C)\}$ which is
$O(N^{-0.53})$ in our
imbalanced settings and $O(N^{-0.56})$ in our balanced settings.
For the other coefficients $O(N^{-1})$
is not ruled out by this argument.
Supplementary Material Figure~S5 reports the mean square errors  for the intercept and the coefficient of the first predictor estimated for different sample sizes, for the three estimators under study in the \Stoo\ setting. We report there the plot only for the first predictor because the mean square errors of the estimates of the seven regression coefficients were essentially equivalent.
All three estimators show a mean square error very
close to $O(N^{-1})$ for $\beta_1$.
Where we anticipated a mean square error no better than
$O(N^{-0.53})$ (imbalanced)
or $O(N^{-0.56})$ (balanced) 
for the intercept
we saw slightly better mean square errors with slopes
between $-0.57$ and $-0.61$, confirming our
expectation that the intercept would be harder
to estimate than the regression coefficients.

\subsection{Variance component estimation}

Here we present the estimation errors in
the variance component parameters $\ssa$
and $\ssb$.  The oracle method is given
the true values of these parameters and so
the comparison is only between all-row-column and
the Laplace method. For the variance parameter $\ssa$, the data only 
have $R$ levels $a_1,\dots,a_R$.  If they were observed
directly, then we could estimate $\ssa$ by 
$(1/R)\sum_{i=1}^R\ai^2$
and have a mean square error of $O(R^{-1})$. 
In practice, the $\ai$ are obscured by the presence
of the signal, the noise $\err_{ij}$ and the
other random effects $\bj$.  Accordingly, the best
rate we could expect for $\ssa$ is 
$O(R^{-1})=O(N^{-\rho})$
and the best we could expect for $\ssb$ is
$O(C^{-1})=O(N^{-\kappa})$.

Figure~\ref{fig:MSE-sigmas-setting211} shows
the mean square error for estimation of $\siga$ and $\sigb$
for all-row-column and Laplace methods in the
\Stoo\ setting.  Due to the imbalance, our
anticipated rates are $O(N^{-0.88})$
for $\siga$ and $O(N^{-0.53})$
for $\sigb$.  The all-row-column method does slightly
better than these rates.  The Laplace method
attains nearly this predicted rate for $\sigb$
but does much worse for $\siga$.
We can understand both of these discrepancies in terms
of biases, described next.
\begin{figure}
\centering
\includegraphics[width=.9\hsize]{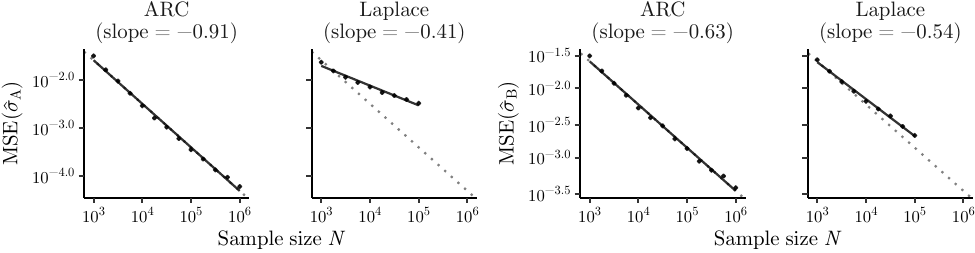}
    \caption{Mean square errors and estimated convergence rates for
    $\siga$ and $\sigb$ for all-row-column and Laplace methods
    in setting \Stoo. {The dotted line in the plot for the Laplace method is the convergence rate of all-row-column.}
In each plot there is a reference
    }
    \label{fig:MSE-sigmas-setting211}
\end{figure}

Figure \ref{fig:boxplots-setting211} 
shows boxplots for the parameter estimates of $\siga$ and $\sigb$ with
the all-row-column and Laplace methods.  We can compare
the center of those boxplots to the reference line
at the true parameter values and see that for the all-row-column estimates there is
a bias decreasing at a faster rate than the
width of the boxes.  This explains the slightly
better than predicted rates that we see for all-row-column. Instead the Laplace method has a substantial bias that only decreases very slowly
as $N$ increases, giving Laplace a worse than
expected rate, especially for $\siga$.

\begin{figure}[!h]
\centering
\includegraphics[width=.9\hsize]{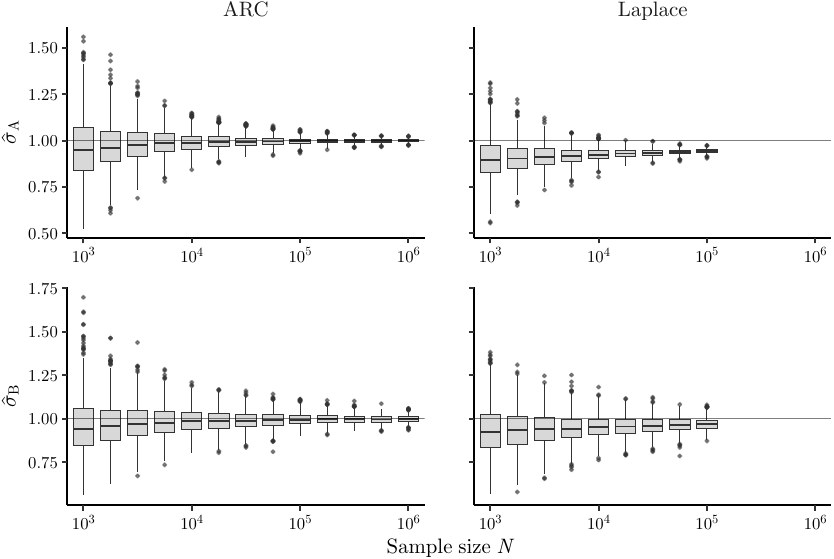}
\caption{Boxplots of $\siga$ and $\sigb$ estimates for the all-row-column and Laplace methods in the setting \Stoo.  There are horizontal reference lines at  the true parameter values.}
\label{fig:boxplots-setting211}
 \end{figure}

\subsection{Other settings}
The simulation results for all eight settings, are reported in full in the Supplementary Figures S1--S24.
We already discussed computational cost in eight settings based on Table~\ref{tab:timings}. 
Here we make brief accuracy comparisons. We compare our proposed all-row-column method to the oracle method which is infeasible because it requires knowledge of $\siga$ and $\sigb$ and
to the Laplace method which becomes infeasible for large $N$ because it does not scale as
$O(N)$.
In most settings and for most parameters, the oracle method was slightly more accurate than all-row-column.  That did not always hold. We see that all-row-column was slightly more accurate than the oracle method for $\beta_1$ in three of the eight settings, namely \Soto\ (Figure {S3}), \Sott\ (Figure {S4}) and \Stto\ (Figure {S7}).

The figures in the Supplementary material
show some cases where all-row-column has a slight advantage over the Laplace method and some
where it has a slight disadvantage.  There
are a small number of cases where the Laplace
method has outliers at $N=10^3$ that cause the
linear regression slope to be questionable.
These are present in the estimates of $\siga$ and $\sigb$ for setting \Stto\ (Figure {S23}).  However the attained mean square errors at $N=10^5$ do not differ much between all-row-column and Laplace
in that setting {(Figure S15)}.  Our conclusion is that compared to the Laplace method, all-row-column is scalable and robust.

The Supplementary Material \S S3 makes a comparison of all-row-column with the maximum pairwise likelihood estimator of \cite{bellio:2005}. 
Their pairwise likelihood involves all the pairs of correlated observations, that is those pairs that share the row- or
 the column-random effect. 
The pairwise likelihood is not computationally attractive.
The number of pairs of data points far exceeds the $O(N)$ constraint.
 In a setting with $R$ rows, the average row has $N/R$ elements in it.  If all
 the rows had that many elements then
 the number of pairs would be
 $\Omega\{R(N/R)^2\}=\Omega(N^2/R)$.
 Unequal numbers of observations per
row can only increase this count.
 As a result the cost must be
$\Omega\{\max(N^2/R,N^2/C)\}=\Omega\{ \max(N^{2-\rho},N^{2-\kappa})\}$, so the
 cost cannot be $O(N)$.

The comparison is made for the \Stoo\ setting that we have been focusing on.  For that imbalanced setting, the
cost of the pairwise likelihood is $\Omega(N^{2-0.53})=\Omega(N^{1.47}).$  We see in \S S3 that the empirical
 cost of the pairwise composite likelihood grows
 as $N^{1.46}$, close to the predicted
 rate.  The pairwise method attains very similar parameter estimation accuracy to all-row-column. The most important difference for this example is that all-row-column costs $O(N)$  while pairwise likelihood is far more
expensive and does not scale to large  data sets.

\section{Application to the Stitch Fix Data}\label{sect:application}

In this section we illustrate our all-row-column method
on a data set from Stitch Fix.  
As described in \cite{ghosh:2021}:
\begin{quote}
``Stitch Fix is an online personal styling service.
One of their business models involves sending
customers a sample of clothing items. The customer may
keep and purchase any of those items and
return the others.
They have provided us
with some of their client ratings data.
That data was anonymized, void of personally identifying
information, and as a sample it does not reflect their total
numbers of clients or items at the time they provided it.
It is also from 2015. While it does not describe their
current business, it is a valuable data set for illustrative purposes.''
\end{quote}

The Stitch Fix data consists of  $N=5{,}000{,}000$
ratings from $R=744{,}482$ clients on $C=3{,}547$ items. The data also includes client- and item-specific covariates.
In this data, the binary response $\yij$
of interest was whether customer $i$ thought that item $j$ was a top rated fit or not,
with $\yij=1$ for an answer of `yes'.
The predictor variables we used are listed
in Table~\ref {tab:variables}. There is one block of client predictors followed by a block of item predictors.
Some of the categorical variable levels in
the data had only a small number of
levels.  The table shows how we
have aggregated them.

\begin{table}
\centering
\begin{tabular}{lll}
{Variable} & {Description} &   {Levels} \\
&&\\
Client fit & client fit profile & {fitted} \\
&& {loose or oversize} \\
&& {straight or tight} \\
{Edgy} & edgy style? &  {yes / no} \\
{Boho} & Bohemian style? &  {yes / no} \\
{Chest} & chest size & numeric \\
{Size} & dress size & numeric \\
&&\\
{Material} & item primary material & {artificial fiber} \\
&& {leather or animal fiber}  \\
&& {regenerated fiber}  \\
&&  {vegetable fiber} \\
{Item fit} & fit of the clothing item &  {fitted} \\
&& {loose or oversized} \\ 
&&  {straight or tight} \\
\end{tabular}
\caption{Predictors available in the Stitch Fix data.}
\label{tab:variables}
\end{table}

Some of the observations with $\zij=1$ were nonetheless incomplete with a few missing entries.
Deleting them left us with $N=4{,}965{,}960$ 
ratings from $741{,}221$ clients on $3{,}523$ items.
The data are not dominated by a single row or column. 
The customer with the most records accounts for $N\epsilon_R$ records where $\epsilon_R\doteq1.25\times 10^{-5}$.  The item with the most records accounts for $N\epsilon_C$ of them with  $\epsilon_C\doteq2.77 \times 10^{-2}$.
The data are sparse because $N/(RC)\doteq 1.9 \times 10^{-3}.$

In a business setting one would fit and compare
a wide variety of different binary regression models
in order to understand the data.  Our purpose here
is to study large-scale probit models
with including crossed random effects
and so we choose just one model
for illustration, possibly the first model of many that one would
consider.
We consider the probit model with crossed random effects whose fixed effects are specified according to the symbolic model formula: 
\begin{center}
Top $\sim$ Client fit + Edgy + Boho + Chest + Size + Material + Item fit,
\end{center}
where Top is the binary response variable 
described earlier.
The model has $p = 12$ parameters for fixed effects including the intercept. The first level of each categorical predictor in alphabetical order (Table~\ref {tab:variables}) is used as the reference level in fitting the model.

Table~\ref{tab:stitch_results} reports 1) the maximum likelihood estimates of the regression parameters under a {na\"ive} probit model that ignores the customer and item heterogeneity, and 2) the all-row-column estimates for the probit model with two crossed random effects for the customers and items. 
The random effects probit parameter estimate $\hat\beta$ equals the na\"ive probit estimate $\hat\gamma$ 
{multiplied} by $(1+\hat\sigma_A^2+\hat\sigma_B^2)^{1/2}$. Using adaptive Gauss-Hermite quadrature we obtained estimates $\hat\sigma_A\doteq0.53$ and $\hat\sigma_B\doteq0.34$.
Following what we learned in the simulation studies, the estimates on all rows and columns were calculated with $k=5$ quadrature nodes. An earlier conservative calculation used $k=28$ but $k=5$ gave indistinguishable results.

The z values reported in the table are computed with the observed information for the {na\"ive} probit model and with the two-way cluster-robust sandwich estimator described in \S\ref{sect:standard_errors}. 
As expected, ignoring the customer and item heterogeneity leads to a large underestimation of the uncertainty in the parameter estimates and thus in the {na\"ive} probit all the predictors are strongly significant given the very large sample size. Conversely, the crossed random effects model takes into account the sources of heterogeneity and leads us to learn that the item fit is not a significant predictor of the top rank and that items made by vegetable fibers are less likely to be ranked top compared to clothes made by artificial fibers.  

\begin{table}[!h]
\centering
\begin{tabular}{llrrrrrr}
&&  \multicolumn{3}{c}{{Na\"ive probit}} & \multicolumn{3}{c}{{Random-effects probit}} \\
{Variable} && Est.${}^\text{\textdagger}$ &z value& p value &  Est.${}^\text{\textdagger}$ & z value & p value\\
  Intercept& & $\phm43.1$ &31.64& $<0.001$& $\phm 50.9$  &10.52 &$<0.001$\\
&&&&&&&\\
Client fit & loose or oversize  &$\phm8.7$ &61.04 &$<0.001$& 10.3 &10.81 &$<0.001$\\
 & straight or tight    & $\phm 5.1$  &34.30&$<0.001$& 6.0 & 10.67 &$<0.001$\\
Edgy  &  yes &$-$3.0 &$-$25.08& $<0.001$& $-$3.5  & $-$7.49 &$<0.001$\\
Boho  & yes & $\phm8.9$ &76.25&$<0.001$ & 10.5 & 25.81 &$<0.001$\\
Chest  & & $-$0.5  &$-$12.30 &$<0.001$ & $-$0.6   & $-$7.32 &$<0.001$\\
Size  & & $\phm0.2$  &10.94 &$<0.001$  & 0.3  & 2.99 & $0.003 $\\
&&&&&&&\\
Material & leather or animal & $-$12.9 &$-$12.99 &$<0.001$ & $-$15.2 & $-$1.57 &${0.116}$ \\
 & regenerated  & $\phm 2.5$ &20.06 &$<0.001$ & $\phm 3.0$ &0.65 & ${0.516}$ \\
  &  vegetable & $-$12.2 &$-$58.39 &$<0.001$ & $-$14.5 & $-$3.13 &$0.002$ \\
Item fit &  loose or oversized & $\phm$9.7 &36.15 &$<0.001$ & $\phm$ 11.4  & 1.78 & {0.075}\\
 & straight or tight & $-2.1$  &$-$9.55 &$<0.001$ & $-$2.5 &$-$0.67 & {0.500} \\
\end{tabular}
\caption{Stitch Fix Binary Regression Results.
${}^\text{\textdagger}$All the estimated predictor parameters are multiplied by 100. 
See Table~\ref{tab:variables} for a description of the predictors. The first level of each categorial predictor in alphabetic order is used as reference level.
}
\label{tab:stitch_results} 
\end{table}

Supplementary Material Figure~S29 compares the two-way cluster robust sandwich standard errors with 1) the standard errors from the {na\"ive} probit fit multiplied by $(1+\hat{\sigma}^2_A+\hat{\sigma}^2_B)^{1/2}\doteq 1.18$ 
to report them in the fixed effects scale of the probit model with crossed random effects
and 2) the pigeonhole nonparametric bootstrap standard errors of \cite{owen:2007} that do not assume correct model specification, as mentioned in \S\ref{sect:standard_errors}. The Figure S29 shows how closely the sandwich and pigeonhole standard errors agree.
The na\"ive standard errors, which ignore dependence between items and customers,  correspond to variances
underestimated by factors ranging from
3 to 954, depending on the
parameter and only slightly on whether
we use sandwich or pigeonhole 
estimates of the coefficient variances.  Thus ignoring the dependencies from correlated data makes an enormous difference here.
These standard errors are reported in 
Table S1 of the Supplementary Material. 
We also computed parametric bootstrap standard errors
(not shown). 
Those were somewhat lower than the nonparametric standard errors as expected.

In an application like that of Stitch Fix data, the typical goal is to make inference about the probability of ranking an item top for a specific customer and a specific item. Such evaluations also require estimating the customer ($a_i$) and item ($b_j$) random effects. The estimates of those random effects are a byproduct of the adaptive Gauss-Hermite quadrature used to approximate the row- and column-wise likelihoods. Figure S30 in the Supplement shows the distribution of the estimated customer and item random effects.

\section{Discussion}\label{sect:conclusions}

In Theorem 2 we proved that there exists a root of the row (or column) likelihood equation that is a consistent estimator of $\tau^2_A$ (or $\tau^2_B).$ This form of consistency is commonly called Cramer consistency. It is the same notion of consistency that \cite{jiang:2013} established for the maximum likelihood estimate. If one does not find Cramer consistency sufficient, then it is possible to construct estimators $\hat\tau^2_A$ and $\hat\tau^2_B$ that converge in probability to $\tau^2_A$ and $\tau^2_B$ as $N\to\infty$. A consistent estimator of $\tau^2_A$ can be obtained from one or more rows $i$ for which  $\nid\to\infty$ and a similar approach works for $\tau^2_B$.
In \S {S1.3} of the Supplementary Material we show how to construct such a consistent estimator  and we  give conditions under which the number of large rows will diverge to infinity as $N\to\infty$. See Theorem 3 there. We prefer our all-row-column estimator to an approach using just large rows, because it would be awkward to have to decide in practice which rows to use, and because we believe that there is valuable information in the other, smaller, rows.

We have used a probit model instead of a logistic one because a Gaussian latent variable is a very natural counterpart to the Gaussian random effects that are the default in random effects models. That connection simplified our modeling and computation. %
\cite{gibb:hede:1997} remark that 
``As in
the case of fixed effect models, selection of probit versus logistic response functions appears to have
more to do with custom or practice within a particular discipline than differences in statistical
properties.'' An extension to logistic regression is outside the scope of this paper.

We are confident that our approach of combining multiple misspecified models will extend to other settings with Gaussian random effects and latent variables.  The code we use already handles ordinal regression.  Extensions to more than two effects or multivariate responses may well work similarly, but are outside the scope of this article.

We conclude with some additional references about recent work on inference for data with a crossed design.
\cite{Goplerud:2023} developed a variational approximation for scalable Bayesian estimation using an appropriate relaxing of the mean-field assumption to avoid underestimation of posterior uncertainty in high dimensions. 

\cite{Xu:2023} combine variational approximations and composite likelihoods that consider row-column decomposition in a similar way to ours. This approach is particularly convenient for Poisson and gamma regression models because in this case analytical calculations allow the approximation of one-dimensional integrals that appear in the composite likelihood to be avoided. 
In the binary case that we considered in our article, the approach of \cite{Xu:2023} requires numerical integration and its consistency has not yet been established. 

\cite{hall:etal:2020} consider message passing algorithms for generalized linear mixed models and their \S 6 includes a crossed effects  binary regression.  However it has $R=10$ and $C=6$ in our notation with $3$ replicates at each $(i,j)$ pair, and it did not address scalability. \cite{Ruli:2016} proposed an improved version of the Laplace approximation to overcome the potential failure of the usual Laplace approximation and also illustrates it in the case of generalized linear models with crossed random effects in their \S 3.5. However, this proposal is not scalable, as illustrated by the numerical results included in
\cite{Ruli:2016}.

\cite{bartolucci2017composite} consider another composite likelihood that combines a row likelihood with a column likelihood to estimate a hidden Markov model for two-way data arrays. This approach shares the philosophy of our proposal but differs from our work in terms of model (latent discrete Markov variables vs crossed random effects), fitting procedure (EM algorithm vs direct maximization), data structure (balanced vs unbalanced and sparse) and motivating  application (genomics vs e-commerce).

\section*{Supplementary Material}

The supplementary material contains proofs of the theorems of \S
\ref{sect:ARC}, additional simulation results discussed in \S
\ref{sect:simu}, further references to the literature, as well as a
table and two plots about the Stitch Fix application mentioned in \S
\ref{sect:application}. The supplement is at\\ 
\url{https://artowen.su.domains/reports/ARC-final-supplement.pdf}.\\
\texttt{R} code for replicating our results is available in the public repository\\ 
\url{https://github.com/rugbel/arcProbit}.

\section*{Acknowledgements}

We are grateful to Silvia Bianconcini, Blair Bilodeau,  Nancy Reid, Alex Stringer, Yanbo Tang and Matt Wand for helpful discussion about generalized linear mixed models, Laplace approximation and adaptive Gauss-Hermite quadrature.
We thank Bradley Klingenberg and Stitch Fix for making this data set available. This work was supported in part by grants IIS-1837931 and DMS-2152780 from the U.S. National Science Foundation, IRIDE from Ca' Foscari University of Venice, and by the research project \emph{Latent Variable Models for Complex Data} founded by the European Union  - NextGenerationEU (MUR DM funds 737/2021).

\bibliographystyle{biometrika}
\bibliography{Bellio-et-al}

\end{document}